\documentclass[11pt]{article}
\usepackage{amsmath,amssymb,amsthm}
\evensidemargin \oddsidemargin
\newtheorem{theorem}{Theorem}
\newtheorem{lemma}{Lemma}
\newtheorem{corollary}{Corollary}
\newtheorem{definition}{Definition}
\newtheorem{assumption}{Assumption}
\title{Fixed Points, a Predictor-Impossibility Theorem, and Applications}
\author{\\Tom Altman\\\\
University of Colorado Denver and Stilman Advanced Strategies\\
Denver, Colorado USA\\
{\tt tom.altman@ucdenver.edu}}
\date{\today\\}
\begin{document}
\maketitle
\begin{abstract}
We introduce an activation hierarchy consisting of stage machines, stage domains, and stage languages generated by an activation operator. The central result is a Predictor-Impossibility Theorem ({PIT}), which shows that no effective predictor family can uniformly determine all stage languages of the hierarchy. The proof combines the semantic activation construction with the {\it S-m-n}\, Theorem and Kleene's Recursion Theorem to obtain a self-referential fixed point that yields a contradiction.

We then define an aggregate language {\it MIS} and establish a slice theorem connecting aggregate inputs to individual stage languages. This provides a bridge from polynomial-time decidability of {\it MIS} to the existence of a predictor family. By PIT, the aggregate language is, therefore, not polynomial-time decidable.

Under the aggregate growth condition defining valid aggregate objects, {\it MIS} is shown to belong to $\mathsf{NP}$. Combining these two results yields {\it MIS} $\in \mathsf {NP \setminus P}$.
The~paper is organized so that PIT stands independently as a recursion-theoretic result, while the complexity-theoretic consequences are derived from the aggregate-language framework.

\end{abstract}

\section{Introduction}

Fixed-point methods occupy a central place in recursion and computability theory. Classical tools such as the {\it S-m-n}\, Theorem and Kleene's Recursion Theorem [1] demonstrate that effective systems can be forced into forms of self-reference that are impossible to avoid. Such techniques underlie many diagonalization arguments and impossibility results throughout recursion theory [2].

This paper studies a hierarchy generated from stage machines through an activation operator. Each stage consists of a machine, a domain, and an associated stage language. The defining feature of the hierarchy is that stage languages are generated from the machines themselves via activation. This self-generated structure creates a natural setting for investigating prediction and self-reference.

The first objective of the paper is recursion-theoretic. We ask whether there exists an effective family of predictors capable of uniformly predicting all stage languages of the hierarchy. The main result is a Predictor-Impossibility Theorem (PIT), which shows that no such predictor family can exist. The proof combines a computable modification of predictor outputs with the {\it S-m-n}\, Theorem and Kleene's Recursion Theorem to obtain a self-referential fixed point. The resulting stage is forced into a contradiction through the interaction of prediction, activation, and diagonal disagreement.

The second objective is complexity-theoretic. Building on the activation hierarchy, we introduce an aggregate language {\it MIS} and establish a Slice Theorem relating aggregate inputs to individual stage languages. This connection provides a bridge from polynomial-time decidability of the aggregate language to the existence of an effective predictor family. By PIT, the aggregate language {\it MIS}\, is, therefore, not decidable in polynomial time [3].

PIT belongs to a broader family of diagonalization-based impossibility results in computability and complexity theory. Unlike classical hierarchy theorems, however, it is formulated in terms of effective predictor families over activation hierarchies.

The paper is organized around the Predictor-Impossibility Theorem. After introducing the activation hierarchy, we establish PIT as an independent recursion-theoretic result. The aggregate language construction is then developed and connected to the hierarchy through the Slice Theorem.
Under the aggregate growth condition defining valid aggregate objects, {\it MIS} is subsequently shown to belong to $\mathsf{NP}$. Its membership there remains stable under a broad range of aggregate-growth assumptions.
Combined together, these two results yield our main theorem {\it MIS} $\in \mathsf{NP \setminus P}$.

\section{Preliminaries and Notation}
Let $n_1 = 2$ and $n_{i+1} = 2^{n_i}$. 
\begin{assumption}[Acceptable Numbering]
The sequence $P_1,P_2,P_3,\ldots$ forms an acceptable numbering of machine descriptions. Thus, the {\it S-m-n}
Theorem and Kleene's Recursion Theorem apply.
\end{assumption}

\subsection{Machine Descriptions}

For a machine $P_i$, the ordinary machine language is denoted by $L(P_i) = {u \in D_i : P_i (u)=1}$.
The notation $L(P_i)$ always refers to the language accepted by the machine itself.

\subsection{Stages}

Each natural number
$i \in \mathbb N$ determines a stage. Associated with stage $i$ are $P_i, D_i, L_i$. The stage domain is $D_i = \{0,1\}^{n_i}$, where $n_i$ is the stage length.
Stage machine $P_i$ is assumed to execute every
input $u \in D_i$ in time $\Theta (n_i ^i)$.

\subsection{Machine Languages and Stage Languages}

A crucial distinction is maintained throughout the paper.
The machine language $L(P_i)$ is the ordinary language accepted by the stage machine. The stage language $L_i$ is generated through the activation operator. The stage language is defined by $L_i = \Phi(P_i)$. This defining relation will be referred to as the \emph{Stage Identity}.

Note that the machine language $L(P_i)$ and the stage language $L_i = \Phi(P_i)$ are distinct in general and should not be identified as the same. This distinction is essential for the fixed-point argument developed in Section 4.

\subsection{Aggregate Objects}

A valid aggregate object in stage $i$ is a tuple $U=(u_1,\ldots,u_m)$ whose components satisfy $u_j \in D_i$ for every $j$. Thus, $|u_j|=n_i$.
The aggregate input size is $q_i = |U| = mn_i$.

\subsection{Distinguished Strings}
For each stage $i$, define $s_i = 0^{n_i}$ as a {\it Distinguished String}. Clearly, $s_i\in D_i = \{0,1\}^{n_i}$.

\subsection{Summary of Notation}
\vspace*{3mm}
\begin{center}
\begin{tabular}{ll}
\hline
Symbol & Meaning\\
\hline
$P_i$ & stage machine\\
$D_i$ & stage domain, $|D_i| \geq 2$\\
$n_i$ & stage length\\
$L(P_i)$ & ordinary machine language\\
$L_i$ & stage language\\
$\Phi$ & activation operator\\
$s_i$ & distinguished stage string $0^{n_i}$\\
$U$ & aggregate object\\
$m$ & number of aggregate components\\
$q_i$ & aggregate input size\\
\hline\\
\end{tabular}
\end{center}
\section{The Activation Hierarchy}
This section introduces the activation hierarchy that serves as the foundation for all subsequent constructions.
\subsection{Stage Machines and Domains}
Each stage machine $P_i$ is total on its stage domain and therefore induces a Boolean-valued function $P_i:D_i \rightarrow \{0,1\}$.
The ordinary machine language is $L(P_i) = \{u\in D_i : P_i(u)=1\}$.
\subsection{Activation}
\begin{definition}[Active Stage]
Stage machine $P_i$ is active if
{\it ACTIVE}$(P_i)=1$ iff\, $\exists u\in D_i$\,\,
$P_i(u)=1$.
\end{definition}

Hence, activity is a semantic property that depends only on the accepted behavior of the machine in its stage domain.
\begin{definition}[Activation Operator]
For every stage machine $P_i$, $\Phi(P_i) = \{v\in D_i : P_i(v)=0 ~\wedge$
{\it ACTIVE}$(P_i)$\}.
\end{definition}

\begin{definition}[Stage Language]

The stage language associated with stage $i$ is $L_i = \Phi(P_i).$
\end{definition}

The hierarchy, therefore, consists of triples $(P_i,D_i,L_i)$.

\subsection{Stage Identity}
The defining relation $L_i = \Phi (P_i)$ will be called the \emph{Stage Identity}. It is a definition, rather than a theorem. Consequently, whenever a stage index $i$ is fixed, the expressions $L_i$ and $\Phi(P_i)$
may be used interchangeably.

\subsection{Activation Complement}
\begin{lemma}[Activation Complement Lemma (ACL)]
For every active stage machine $P_i$, we have
$\Phi(P_i) = D_i \setminus L(P_i) =$ 
$\overline{L(P_i)}$ relative to the stage domain.
\end{lemma}

\begin{proof}
Assume {\it ACTIVE}$(P_i) = 1$.
Then $\Phi (P_i) = \{v \in D_i : P_i(v) = 0\}$.
Since $L(P_i) = \{v \in D_i : P_i(v) = 1\}$
and $P_i$ is total on $D_i$, the two sets are complementary within the stage domain.

Hence, $\Phi(P_i) = D_i \setminus L(P_i) = \overline{L(P_i)}$.
\end{proof}

\section{Predictor Impossibility Theorem (PIT)}
We now establish the central recursion-theoretic result of the paper.
\subsection{Predictors}
\begin{definition}[Predictor Family]
A predictor family is an effective family of machines $G(i)$ satisfying
$L(G(i)) = L_i$ for every stage $i$.
\end{definition}

Thus, a predictor family correctly determines the stage language associated with every stage of the hierarchy.
\subsection{Index Realization}
Given a predictor family $G$, define a machine transformation
$F(i) = G(i)\,\cup\,\{s_i\}$.
Equivalently, 
\vspace*{3mm}\\
$F(i)(x) = 
\begin{cases}
1, & x=s_i,\\
G(i)(x), & x\neq s_i.
\end{cases}$

\begin{lemma}[Index Realization]
There exists a computable index transformation
$f:\mathbb N\rightarrow\mathbb N$ such that $P_{f(i)} = F(i)$
for every stage $i$.
\end{lemma}

\begin{proof}
Since $G$ is effective, a description of $(G(i))$ can be computed from $i$.
The distinguished string $s_i=0^{n_i}$ is computable from the stage index.
Hence, a procedure exists that constructs machine
\vspace*{3mm}\\
$x\,\mapsto
\begin{cases}
1,& x=s_i,\\
G(i)(x),& x\neq s_i.
\end{cases}$
\vspace*{3mm}

The existence of a computable index transformation $f$ follows from the {\it S-m-n}\, Theorem.
\end{proof}
\subsection{Predictor Impossibility}
\begin{theorem}[Predictor Impossibility Theorem (PIT)]
No effective predictor family exists.
\end{theorem}

\begin{proof}
Assume that a predictor family $G$ exists. Let $f$ be the index transformation obtained in the Index Realization Lemma. By Kleene's Recursion Theorem there exists an index $i^*$ such that
$P_{i^{*}} = P_{f(i^{*})}$. Consequently, $L(P_{i^{*}}) = L(P_{f(i^{*})})$.
Using the definition of $f$, $L(P_{i^{*}}) = L(G(i^{*})) \cup \{s_{i^*}\}$.
Using the predictor property, $L(G(i^{*})) = L_{i^{*}}$.
Therefore, it follows that $L(P_{i^{*}}) = L_{i^{*}} \cup \{s_{i^{*}}\}$.

Applying Stage Identity, $L_{i^{*}} = \Phi(P_{i^{*}})$, yields
$L(P_{i^{*}}) = \Phi(P_{i^{*}}) \cup \{s_{i^{*}}\}$.
Since $s_{i^{*}} \in L(P_{i^{*}})$, there exists $u\in D_{i^{*}}$
such that $P_{i^{*}}(u)=1$.
Thus, {\it ACTIVE}$(P_{i^{*}})=1$.
The Activation Complement Lemma, thus, gives
$\Phi(P_{i^{*}}) = \overline{L(P_{i^{*}})}$.
Then $L(P_{i^{*}})$ = 
$\overline {L(P_{i^{*}})} \cup \, \{s_{i^{*}}\}$.
Pick any $x\in D_{i^{*}}, ~x \neq s_{i^{*}}$. We~get
$x\in L(P_{i^{*}}) \iff x \in \overline{L(P_{i^{*}})}$. Contradiction. Hence no effective predictor family exists.
\end{proof}
\section{The Aggregate Language {\emph {MIS}}}
The Predictor-Impossibility Theorem is independent of complexity theory. We now introduce an aggregate language whose polynomial-time decidability would imply the existence of an effective predictor family, thus connecting the recursion-theoretic obstruction to computational complexity.
\subsection{Aggregate Acceptance}
For each stage $i$, let $z_i \in D_i$,
where $L_i(z_i) = 0$, denote an inert sentinel associated with that stage.
\begin{definition}[Aggregate Language]
Let $U=(u_1,\ldots,u_m)$
be a valid aggregate object belonging to stage $i$.
The aggregate language {\it MIS} is defined by {\it MIS}$(U) = 1$
\,iff\, $\exists \,j\in \{1,\ldots,m\} \quad L_i(u_j) = 1$.
\end{definition}
Thus, an aggregate object is accepted exactly when at least one of its components belongs to the corresponding stage language.
\subsection{Slice Theorem}
Every stage language occurs as a distinguished slice of the aggregate language.
\begin{theorem}[Slice Theorem]
For every stage $i$,  {\it MIS}$(u,z_i,\ldots,z_i) = L_i(u)$.
\end{theorem}

\begin{proof}
By definition of aggregate acceptance, the aggregate object
$(u,z_i,\ldots,z_i)$ is accepted iff the distinguished component $u$ belongs to the stage language.
Hence, {\it MIS}$(u,z_i,\ldots,z_i) = L_i(u)$.
\end{proof}
\subsection{Bridge Theorem}

The Slice Theorem converts a polynomial-time decider for ${\it MIS}$ into an effective predictor family.
\begin{theorem}[Bridge Theorem]
If {\it MIS} $\in \mathsf{P}$, then an effective predictor family exists.
\end{theorem}

\begin{proof}
Assume {\it MIS} $\in \mathsf{P}$. Let $M$ be a polynomial decider for {\it MIS}. Define $G(i)(u) = M(u,z_i,\ldots,z_i)$.
By the Slice Theorem, $L(G(i)) = L_i$ for every stage $i$.
Thus $G$ is an effective predictor family.
\end{proof}
\subsection{Complexity Consequence}
\begin{corollary}
{\it MIS} $ \notin \mathsf{P}$.
\end{corollary}

\begin{proof}
Assume {\it MIS} $\in \mathsf{P}$.
Then, by the Bridge Theorem, an effective predictor family exists.
This~contradicts the Predictor-Impossibility Theorem.
\end{proof}
\section{Membership of {\emph {MIS}} in $\mathsf{NP}$}
\subsection{Aggregate Growth}
Let $U=(u_1,\ldots,u_m)$ be a valid aggregate object belonging to stage $i$.
Each component satisfies $|u_j| = n_i$.
The aggregate input size is $q_i = |U| = mn_i$.
Valid aggregate objects satisfy the aggregate growth condition
$m = \Omega(n_i^{i-k})$ for some fixed constant $k$,
independent of the input. Consequently, $q_i = \Omega(n_i^{i-k+1})$.
\subsection{Witness Verification}
By definition, {\it MIS}$(U)=1$ if and only if there exists a component $u_j$ such that $L_i(u_j)=1$.
A~successful component therefore serves as a witness.
Given $(U,u_j)$, a verifier checks
that $u_j$ is one of the components of $U$ and
that $L_i(u_j)=1$.

The first step requires $O(q_i)$ time.
The second requires $\Theta(n_i^i)$ time.
Since $q_i = \Omega(n_i^{i-k+1})$,
we obtain $n_i = O\left(q_i^{1/(i-k+1)} \right)$.
Hence $n_i^i = O\left(q_i^{i/(i-k+1)} \right)$.
Since $\frac{i}{i-k+1} = 1+\frac{k-1}{i-k+1} \le k$,
the exponent is bounded by a constant independent of the input.
Thus, $n_i^i = O(q_i^k)$, which is polynomial in $q_i$.

\subsection{$\mathsf{NP}$ Membership}
\begin{theorem}[$\mathsf{NP}$ Membership]
{\it MIS} $\in \mathsf{NP}$.
\end{theorem}

\begin{proof}
The witness is a successful component $u_j$.
By the preceding analysis, the witness can be verified in time
$O(q_i^k)$, which is polynomial in the size of the aggregate input.
Thus, {\it MIS} $\in \mathsf{NP}$.
\end{proof}
The proof {\it MIS} $\in \mathsf{NP}$ uses the aggregate growth condition
$m=\Omega(n_i^{i-k})$ for a fixed constant $k$.
This condition can be weakened.
Define {\it MIS}$^{\,(r)}$ to be the restriction of {\it MIS} to aggregate objects satisfying $m=\Omega(n_i^{i-r})$.
For every fixed $r$, {\it MIS}$^{\,(r)}\in \mathsf{NP}$.
The proof is identical.
Consequently, $\mathsf{NP}$ membership remains stable under a broad range of aggregate-growth assumptions.
\section{Main Theorem}
\begin{theorem}[Main Theorem]
{\it MIS} $\in \mathsf{NP \setminus P}$.
\end{theorem}

\begin{proof}
By Theorem 4, {\it MIS} $\in \mathsf{NP}$ and {\it MIS} $\notin \mathsf{P}$ by Corollary 1, so {\it MIS} $\in \mathsf{NP \setminus P}$.
\end{proof}

\subsection{Dependency Structure}
The logical dependencies $(\Downarrow)$ and direct implications $(\rightarrow)$ in the paper are summarized below:\vspace*{5mm}

\text{Acceptable Numbering + {\it S-m-n}\, Theorem + Kleene Recursion Theorem}

$\Downarrow$

\text{Predictor-Impossibility Theorem 1 + Slice Theorem 2 + Bridge Theorem 3}

$\Downarrow$

\text{Corollary 1} $\rightarrow$ {\it MIS} $\notin \mathsf{P}$ 

$\Downarrow$

\text{Theorem 4} \,$\rightarrow$ {\it MIS} $\in \mathsf{NP}$

$\Downarrow$ 

\text{Theorem 5} \,$\rightarrow$ {\it MIS} $\in \mathsf{NP \setminus P}$.


\begin{thebibliography}{9}
\bibitem{Kleene}
C. Kleene, \emph{Introduction to Metamathematics}, North-Holland, 1952.
\bibitem{Rogers}
H. Rogers, Jr., \emph{Theory of Recursive Functions and Effective Computability}, MIT Press, 1987.

\bibitem{Sipser}
M. Sipser, \emph{Introduction to the Theory of Computation}, 3rd  ed., Cengage Learning, 2012.
\end{thebibliography}
\end{document}